\documentclass[journal]{IEEEtran}

\usepackage[cmex10]{amsmath}
\usepackage{mathrsfs}
\usepackage{cite}
\usepackage[colorinlistoftodos]{todonotes}
\usepackage{enumerate}

\usepackage{amssymb,amsthm}
\usepackage{tikz}
\usetikzlibrary{calc}%,intersections,automata,arrows}
\usepackage{pgfplots}
\pgfplotsset{compat=1.14}
\usepackage[utf8]{inputenc}

\newtheorem{theorem}{Theorem}

\newtheorem{lemma}{Lemma}

\theoremstyle{definition}

\theoremstyle{remark}

\begin{document}
\allowdisplaybreaks

\title{
Communication with Crystal-Free Radios
}

\author{Dor Shaviv, Ayfer \"{O}zg\"{u}r, and Amin Arbabian
\thanks{This work was supported 
in part by a Robert Bosch Stanford Graduate Fellowship, 
in part by the National Science Foundation under Grant CCF-1618278, 
% in part by the Center for Science of Information, an NSF Science and Technology Center, under Grant CCF-0939370,
and in part by the SystemX Alliance.
This work was presented in part at the 2017 IEEE Global Communications Conference (GLOBECOM)~\cite{crystalfreeGLOBECOM}.
}
\thanks{The authors are with the Department of Electrical Engineering, Stanford University, Stanford, CA 94305 USA (e-mail: shaviv@stanford.edu; aozgur@stanford.edu, arbabian@stanford.edu).}
}

\maketitle

\begin{abstract}
We consider a communication channel where there is no common clock between the transmitter and the receiver. This is motivated by the recent interest in building system-on-chip radios for Internet of Things applications, which cannot rely on crystal oscillators for accurate timing. We identify two types of clock uncertainty in such systems: timing jitter, which occurs at a time scale faster than the communication duration (or equivalently blocklength); and clock drift, which occurs at a slower time scale. We study the zero-error capacity under both types of timing imperfections, and obtain optimal zero-error codes for some cases. Our results show that, as opposed to common practice, in the presence of clock drift it is highly suboptimal to try to learn and track the clock frequency at the receiver; rather, one can design codes that come close to the performance of perfectly synchronous communication systems without any clock synchronization at the receiver.
\end{abstract}


\IEEEpeerreviewmaketitle

\section{Introduction}
\label{sec:introduction}

The next exponential growth in connectivity is projected to be no longer in access between people but in connecting objects and machines in the age of Internet of Things (IoT). This is partly fueled by the emergence of tiny, low-cost wireless devices that combine communication, computation and sensing. These wireless devices are expected to form the fabric of smart technologies and cyberphysical systems, enabling a plethora of exciting applications: from in-body and personal health monitoring, to smart homes and transportation systems, to automation and monitoring in smart grids.

However, scaling wireless devices from billions to potentially trillions (as envisioned by some forecasts \cite{bassi2008internet,trillion3}) requires orders of magnitude reduction in costs and often size, both of which are dominated by external components such as batteries and crystal oscillators. This has led to significant recent interest in building miniature radios that do not possess any external components \cite{papotto201490,bhamra201524,bae201545,sebastiano2009low,drago2009impulse,tabesh2015power}.
For example, the ant-radios of \cite{tabesh2015power} integrate a full wireless communication system, including the full transceiver, antenna, and clock, on a single CMOS chip of size $4.4\  \text{mm}^2$.  A small crystal oscillator, on the other hand, is around $1.9\ \text{mm}^2$, which is about half the size of the entire system. 
%For this reason, and in order to make the radio completely stand-alone by eliminating extra steps for integration/packaging/assembly (and thus low-cost), 
In addition to reducing size and cost, eliminating external components is also desirable for eliminating the extra steps for integration, packaging, and assembly. In particular, the ant-radios of \cite{tabesh2015power} use
an on-chip low-power and low-accuracy 200 MHz ring oscillator to control the symbol rate instead of a crystal oscillator, and operate without a battery; they are instead powered remotely via wireless power transfer.
%The radio also does not  have a traditional battery but is powered externally via wireless power transfer.
%Furthermore, the radio operates without a battery, but is instead powered remotely via wireless power transfer.

Compared to crystal resonator-based systems, ring oscillator systems experience greater jitter and drift, causing the clock frequency variation to lie within a $\sim 100$ MHz range centered at 200 MHz. This is incompatible with many conventional communication schemes and poses a significant design challenge. Normally, the receiver employs a timing recovery mechanism such as an early-late gate \cite{proakis2008digital} to extract the transmitter's clock (or symbol rate). However, this is only possible when the transmitter's clock is relatively stable.

\begin{figure}
\centering
\def \picheight {6.5pc}
\includegraphics[height=2cm]{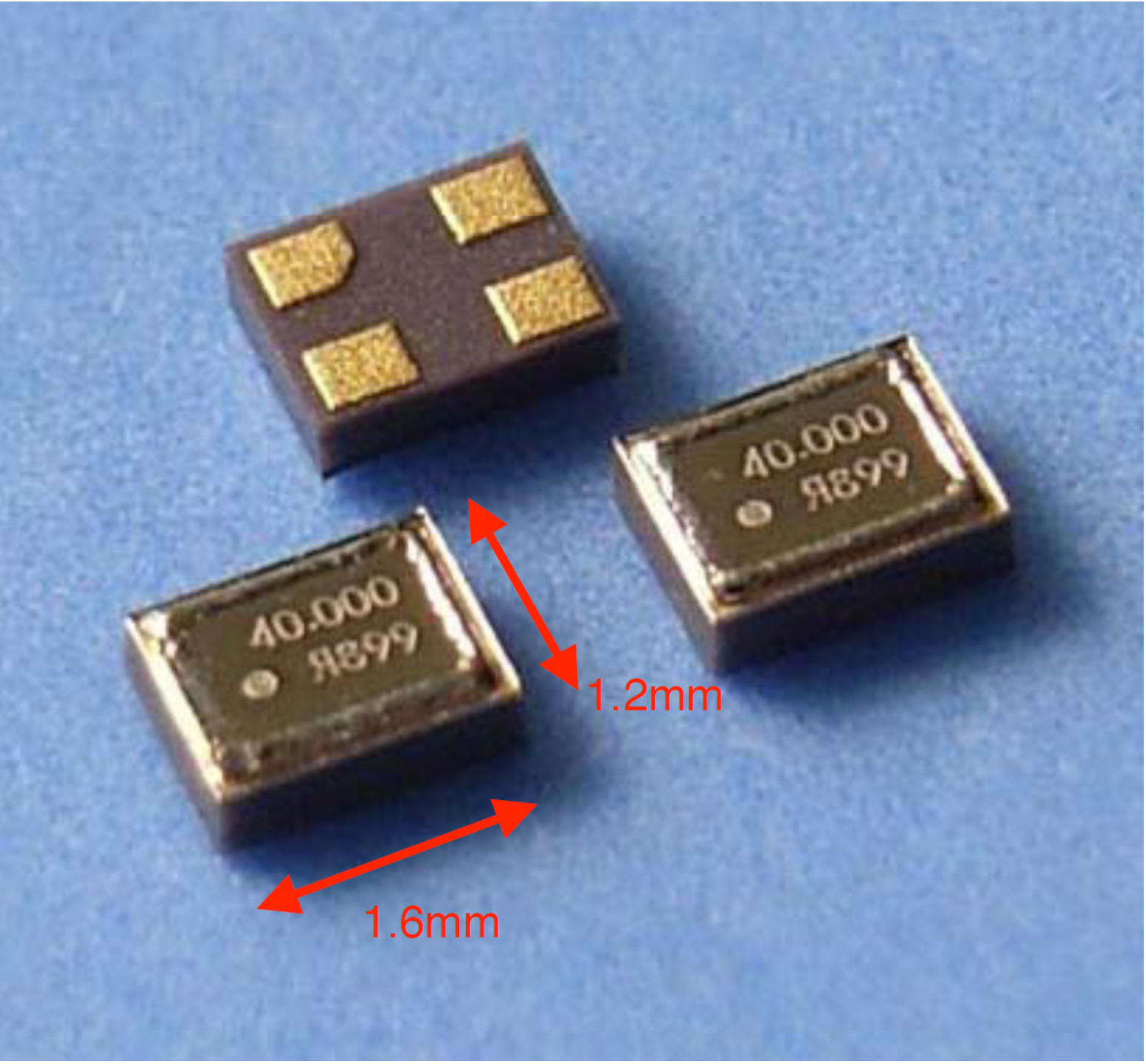}\ 
\includegraphics[height=2cm]{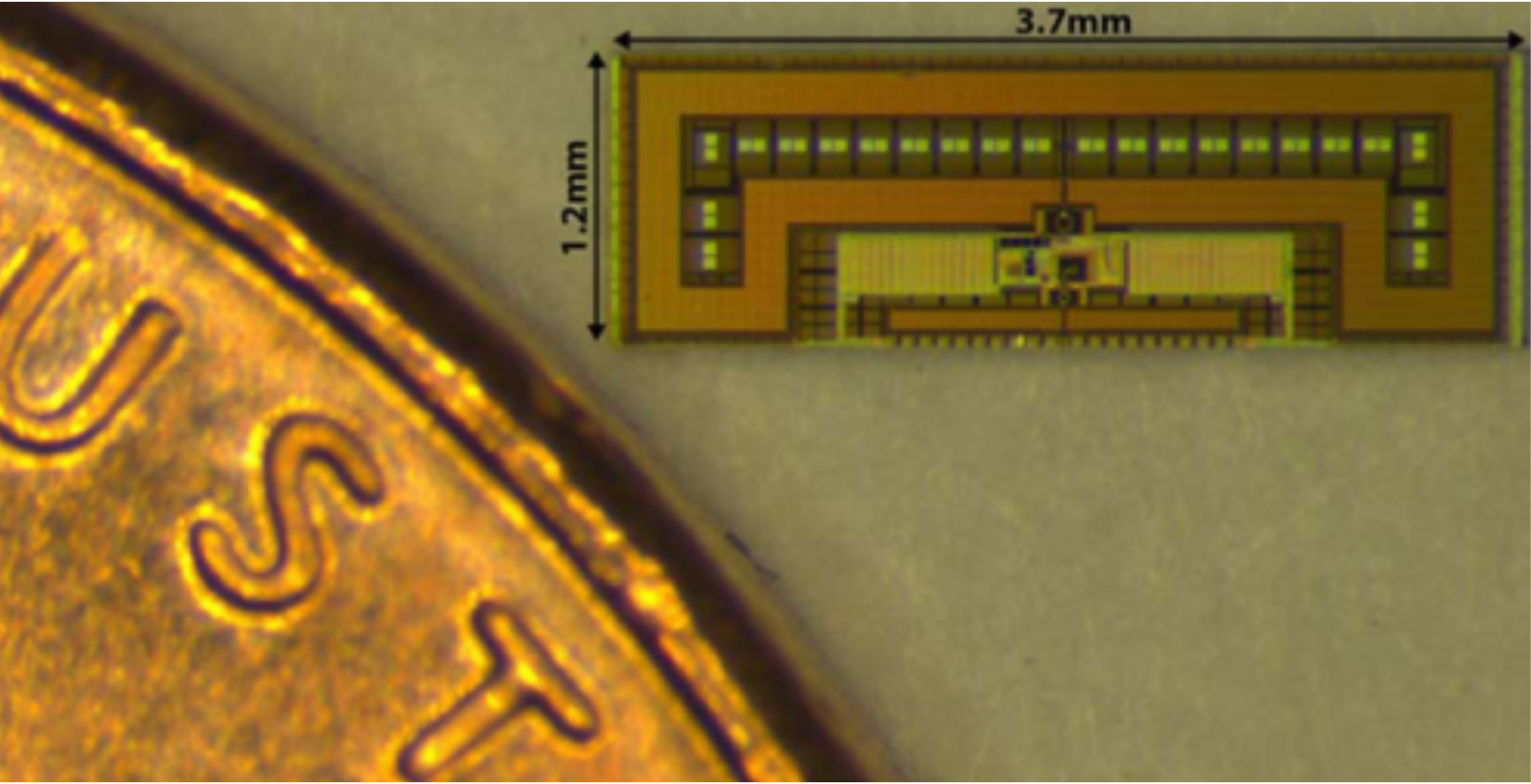}
\caption{(a) A typical crystal oscillator. (b) The ant-radios of \cite{tabesh2015power}.}
\end{figure}

In particular, consider transmission using pulse-position modulation (PPM), as done in \cite{tabesh2015power} due to the energy efficiency of this modulation technique for wideband communication. In PPM, information is encoded in the position of a pulse transmitted in one of $M$ bins, where $M=64$ in \cite{tabesh2015power}.
The bin duration is determined by the inaccurate ring oscillator, and can vary between 4 and 7 ns. Thus the uncertainty in the transmitter's clock makes it impossible for the receiver to decode the received message. To overcome this problem, in \cite{tabesh2015power} transmission begins with two extra pulses, transmitted back to back in two consecutive bins, and the frame size $M$ is restricted to $64$. The receiver can learn the bin duration (and thus the transmitter's clock) by measuring the time between the first two pulses, and subsequently decode the location of the third (information-bearing) pulse.\footnote{The width of the transmitted pulse is much shorter than the duration of the bin in \cite{tabesh2015power}.}
Restricting the frame size limits the amount of accumulated jitter and prevents the transmitter and receiver clock from going out of sync during the course of transmission. This synchronization cost presents a significant burden on the transmitter, as the energy consumption of the transmitter is dominated by the transmitted energy, $2/3$ of which is now spent on synchronization. %See Fig.~\ref{fig:amin_scheme}.

The current paper provides a study of reliable communication in such systems, where there is no common clock between the transmitter and the receiver, from a fundamental perspective.
Motivated by digital recording, communication without a synchronous clock has been considered in previous information theoretic literature \cite{yeung2009reliable,baggen1993information,shamai1991bounds,hekstra1993capacity,iyengar2016capacity}, where these works model the absence of a common clock as timing jitter.
For example in \cite{yeung2009reliable}, which is most closely related to our work, jitter causes the transmitted signal to be arbitrarily ``stretched'' or ``squeezed'' in time by a varying factor during the course of communication. In other words, the ``stretching'' or ``squeezing'' occurs at a time scale faster than the duration of communication (or equivalently blocklength). In \cite{tabesh2015power} however, the clock remains sufficiently stable during the course of the 64-PPM symbol. The real challenge is that each time the transmitter sends a 64-PPM symbol, it is encoded with an unknown (but stable) clock whose frequency can lie anywhere between $150$ and $250$ MHz. See Fig.~\ref{fig:amin_scheme}.

\begin{figure}
\centering
\begin{tikzpicture}
\tikzstyle{every node}=[font=\scriptsize];

% Parameters
\def \framelength {2.5};
\def \framegap {2};
\def \pulseheight {1};

% Axis
\draw (0,0) -- (\framelength,0);
\draw[dashed] (\framelength,0) -- (\framelength+\framegap,0);
\draw (\framelength+\framegap,0) -- (2*\framelength+\framegap,0);

% First frame
\foreach \x in {0.1, 0.3, 1.1}
{
  \draw[line width=2pt] (\x,0) -- (\x,\pulseheight);
}
\foreach \x in {0.1,0.3,...,\framelength}
{
  \draw (\x,-0.1) -- (\x,0.1);
}

% Second frame
\foreach \x in {0.1, 0.45, 1.5}
{
  \draw[line width=2pt] (\framelength+\framegap+\x,0) -- (\framelength+\framegap+\x,\pulseheight);
}
\foreach \x in {0.1,0.45,...,\framelength}
{
  \draw (\framelength+\framegap+\x,-0.1) -- (\framelength+\framegap+\x,0.1);
}

\end{tikzpicture}
\caption{Transmission scheme in \cite{tabesh2015power}. Note that the common topology envisioned for IoT applications is that a large number of transmitters access a single sink node. Therefore, successive transmissions of 64-PPM symbols from a given transmitter are interleaved by large time intervals due to a TDMA scheme between a large number of transmitters.}
\label{fig:amin_scheme}
\end{figure}
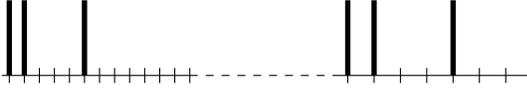

Therefore, in this work, we distinguish between two types of clock uncertainty at the receiver: \emph{timing jitter}, which can cause the transmitter's clock to vary arbitrarily during the course of transmission; and \emph{clock drift}, which occurs at a time scale much larger than the blocklength. The second can be modeled as a fixed but unknown clock. 
Timing jitter was studied in \cite{yeung2009reliable}, where capacity was found and optimal codes were developed.
However, the optimal codes for communication under clock drift is fundamentally different, and in this work we aim to develop codes that are optimal when both imperfections are present. 

We show for example that when only clock drift is present, it is possible to code in such a way that the receiver never learns the exact clock frequency: by considering \emph{ratios} of pulse positions instead of their absolute values, the clock cycle indeed does not play a part. This could be used to almost entirely eliminate the cost of synchronization in \cite{tabesh2015power} (the extra two pulses used to convey the transmitter's clock to the receiver).
Indeed, we show that our scheme can improve from a rate of 6 bits per frame obtained by the 64-PPM scheme, to a rate of 10.76 bits per frame by encoding over ratios, nearly approaching the perfect synchronization upper bound which is 11.02 bits per frame.

%Our model most closely resembles the model in \cite{yeung2009reliable}, in which jitter causes the transmitted signal to be arbitrarily ``stretched'' or ``squeezed'' in time at the receiver side. 
%
%Timing jitter has been studied previously mostly in the context of digital recording \cite{baggen1993information,shamai1991bounds,hekstra1993capacity}, and was modeled e.g. by letting `1' symbols shift left or right with some probability, or as a deletion and insertion channel \cite{iyengar2016capacity}.
%In this work, following the approach in \cite{yeung2009reliable}, we tackle the problem from a different angle by considering the zero-error capacity of the channel without imposing a probabilistic model.

\section{Channel Model}
\label{sec:model}

We consider {multi-pulse PPM} communication where the transmitter sends $k$ short (0-width) pulses in $M$ \emph{bins}, where each pulse is located in one bin and information is encoded in the position of the pulses (or equivalently the occupied bins). Each of the $\binom{M}{k}$ possible transmit signals can be represented as a binary sequence of length $M$, where $1$ indicates the presence of a pulse in the corresponding bin.
Instead of this, however, in this paper we adopt an equivalent differential representation of the signals, where each one is represented by a vector of length $k$, $(X_1,\ldots,X_k)$, where $X_i$ is the time (number of clock cycles, or number of bins) between the $(i-1)$-th and $i$-th pulses, which is also called the $i$-th \emph{run}. Note that the first run $X_1$ is simply the bin of the first pulse (equivalently define $X_0=0$).
The runs $X_i$ take values in the set $\{1,\ldots,M\}$, and the vector must satisfy $\sum_{i=1}^{k}X_i\leq M$, since there are exactly $k$ pulses in the transmitted signal.
Let the set of all such legitimate input vectors be denoted by~$\mathcal{X}$. When $k$ is clear from the context, we will use boldface $\mathbf{X}$ as a shorthand for the vector $(X_1,\ldots,X_k)$. In this paper we would like to study zero-error communication with vectors from $\mathcal{X}$ in the presence of clock imperfections as we model next. 
While in this paper we only focus on communication with {multi-pulse PPM} (both for simplicity and because this is the modulation of choice for most low-energy systems), our model and results can be extended to allow general modulation techniques in the direction of \cite{yeung2009reliable} {(e.g. pulse-code modulation)}. 

Note that we will keep the blocklength $M$ finite here as it is typically not a large number for systems of interest. We are interested in understanding the structure of optimal codes and the size of the optimal code for finite $M$, rather than the behavior of capacity as $M$ gets large. 
{Moreover, the problem trivializes for $M\to\infty$: if $k$ remains finite, the rate is zero; on the other hand, if $k$ grows with $M$, then the first two pulses can be used to perfectly learn the clock as in \cite{tabesh2015power} without any loss in the communication rate, and the problem becomes identical to one with perfect synchronization.}

By using input vectors from $\mathcal{X}$, our goal is to achieve zero-error communication under the presence of the following two types of clock imperfections:

\paragraph*{Clock drift} The receiver observes the transmitted vector multiplied by an unknown fixed real number $T$ that takes values in a closed interval $[T_1,T_2]$, for $0<T_1\leq T_2$. We will also be interested in the case of unbounded clock drift, such that $T\in[T_1,\infty)$. Hence the observed vector is $T\mathbf{X}$, i.e. the observed run lengths are given by $TX_i$, $i=1,\ldots,k$. Note that this models the scenario where the receiver is unaware of the clock used by the transmitter (it only knows that it lies in a certain interval), but the transmitter's clock remains stable during the transmission of the signal. This models variations of the transmitter's clock frequency at a scale larger than the blocklength for communication (in a flavor similar to large scale fading in wireless systems \cite{Davidsbook}).

\paragraph*{Timing jitter} On top of the slow clock drift, the transmitter's clock experiences random jitter, i.e. variations at a scale faster than the blocklength (in a flavor similar to small scale fading in wireless systems \cite{Davidsbook}). We model this similarly to~\cite{yeung2009reliable} by a strictly positive arbitrary process $r(u)$, unknown to the transmitter nor to the receiver, such that $r(u)\in[a,b]$ for some $0<a\leq b<\infty$. This process represents the instantaneous deviation of the clock from its nominal frequency. If a pulse is transmitted at time $t$, the receiver observes a pulse at time $\int_0^tr(u)du$. Thus, the runs observed at the receiver are given by
\[
Y_i=\int_{\sum_{j=1}^{i-1}TX_j}^{\sum_{j=1}^{i}TX_j}r(u)du,
\qquad i=1,\ldots,k.
\]
Equivalently, we can write
\begin{equation}
Y_i = T Z_i X_i,\qquad i=1,\ldots,k,
\end{equation}
where the $Z_i$'s are arbitrary, independent of each other, and $Z_i\in [a,b]$.
See Fig.~\ref{fig:signal_example} for an illustration of transmitted and received signals.
% This comment does not make sense to me: Note that communication is \emph{noiseless} in the sense that only the pulse arrival times are changed by the channel and not their amplitudes.

\begin{figure}
\centering
\begin{tikzpicture}
\tikzstyle{every node}=[font=\scriptsize];

% Parameters
\def \pulseheight {1};
\def \y {-2.5};
\def \T {1.5};

% Transmitted signal
\draw (0,0) -- (6.5,0);
\foreach \x in {0,0.5,...,6.5}
{
  \draw (\x,-0.1) -- (\x,+0.1);
}
\foreach \xl/\xr/\i in {0/1/1,1/2/2,2/4/3}
{
  \draw[line width=3pt] (\xr,0) -- (\xr,\pulseheight);
  \draw[<->] (\xl+0.05,-0.2) -- node[below] {$X_{\i}$} (\xr-0.05,-0.2);
}
\node at (-0.5,0.5*\pulseheight) {(a)};

% Signal after clock drift
\draw (0,\y) -- (6.5,\y);
\foreach \x in {0,0.5,...,4}
{
  \draw (\T*\x,\y-0.1) -- (\T*\x,\y+0.1);
}
\foreach \xl/\xr/\i in {0/1/1,1/2/2,2/4/3}
{
  \draw[line width=3pt] (\T*\xr,\y) -- (\T*\xr,\y+\pulseheight);
  \draw[<->] (\T*\xl+0.05,\y-0.2) -- node[below] {$TX_{\i}$} (\T*\xr-0.05,\y-0.2);
}
\node at (-0.5,\y+0.5*\pulseheight) {(b)};

% Received signal (clock drift and jitter)
\draw (0,2*\y) -- (6.5,2*\y);
\foreach \xl/\xr/\i in {0/0.8/1,0.8/2.2/2,2.2/3.7/3}
{
  \draw[line width=3pt] (\T*\xr,2*\y) -- (\T*\xr,2*\y+\pulseheight);
  \draw[<->] (\T*\xl+0.05,2*\y-0.2) -- node[below] {$TZ_{\i}X_{\i}$} (\T*\xr-0.05,2*\y-0.2);
}
\node at (-0.5,2*\y+0.5*\pulseheight) {(c)};

\end{tikzpicture}
\caption{Example of transmitted signal and received signal.
(a) Transmitted signal.
(b) Signal after the effect of clock drift.
(c) Received signal $\mathbf{Y}$.}
\label{fig:signal_example}
\end{figure}
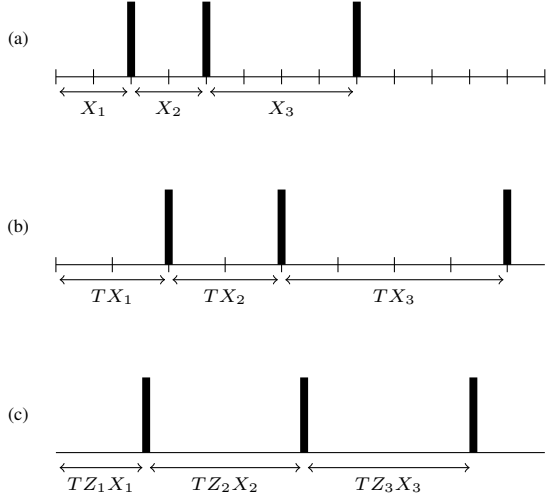

%This communication model was inspired by \cite{yeung2009reliable}, but deviates from it in several ways.
%First, we model communication using short pulses instead of step functions. It can be seen that this is equivalent to \cite{yeung2009reliable} if the step function takes values in a binary alphabet, and so information is only encoded in the times between pulses (runs). This simplifies the exposition while still extracting the interesting insights from the model.
%Similarly, we consider only integer runs, as opposed to possibly continuous runs in \cite{yeung2009reliable}; this is merely a technical difference and does not impact the theory developed here.
%
%Second, it is assumed here that the receiver re-synchronizes at the beginning of each frame. This is done, for example in \cite{tabesh2015power}, by sending a synchronization marker from the receiver to the transmitter at fixed intervals, which determines the start of each frame. This is in contrast to \cite{yeung2009reliable}, in which the receiver re-synchronizes \emph{implicitly} after decoding each run. As will be seen in the next section, this is not possible in the presence of slow clock drift, since the receiver must wait for several runs in order to be able to decode.

A pair of input vectors $\mathbf{x}$ and $\mathbf{x}'$ are said to be \emph{indistinguishable} at the output if there exist $T,T'\in[T_1,T_2]$ and $Z_1,\ldots,Z_k,Z'_1,\ldots,Z'_k\in[a,b]$ such that 
\begin{equation}
T Z_i x_i = T' Z'_i x'_i,\qquad i=1,\ldots,k.
\label{eq:def_indistinguishable}
\end{equation}
That is, two input vectors are indistinguishable if they can produce the same signal at the output.
Accordingly, two vectors are \emph{distinguishable} if they are not indistinguishable.
{As observed in~\cite{yeung2009reliable},}
it can be seen from \eqref{eq:def_indistinguishable} that, rather than the actual boundaries of the intervals $[T_1,T_2]$ and $[a,b]$, only their ratios are relevant. Therefore we are motivated to define the quantities 
\[\xi = \frac{b}{a}\quad\text{and}\quad\gamma = \frac{T_2}{T_1}.\]
Note that $\xi,\gamma\geq1$, where equality means the absence of jitter or clock drift, respectively. Note also that $\gamma$ can be infinity.

A zero-error code $\mathcal{C}$ is a set of input vectors, called \emph{codewords}, such that all of them are distinguishable at the receiver. 
{Note that there is no notion of probability here; the codewords are required to be distinguishable for any possible realization of $T$ and $Z_i$'s.}
We say that the rate of a code is $R={\log|\mathcal{C}|}$ bits per frame, and the maximum of all rates is the zero-error capacity.
In the following sections, we study optimal zero-error codes for the channel defined by $(k,M,\xi,\gamma)$, which are zero-error codes with the maximal number of codewords.

Note that the classical definition of zero-error capacity \cite{shannon1956zero} concerns the maximal rate achieved asymptotically over many repeated uses of the channel; in this work, we consider a different notion of capacity by defining codes for a single block of finite fixed length (one-shot).

\section{Optimal Codes}

In this section we study optimal codes for the channel $(k,M,\xi,\gamma)$, for several special cases of interest.
Denote an optimal code by $\mathcal{C}^*_{\xi,\gamma}$, where $k,M$ should be understood from the context.

\subsection{No Jitter ($\xi=1$) and Unbounded Clock Drift ($\gamma=\infty$)}
The clock drift is $T\in[T_1,\infty)$ for some $T_1>0$, and we can let $Z_i=1$ without loss of generality.
First, observe that if $k=1$, reliable communication is not possible. This is because upon transmitting $X_1$, the output $Y_1$ can be any number in $[T_1X_1,\infty)$, and all input signals are indistinguishable.

Assume $k\geq2$. For an input vector $\mathbf{x}=(x_1,\ldots,x_k)$, define the \emph{ratios vector} $\mathbf{u}=(u_2,\ldots,u_{k})$ by
\[
u_i=\frac{x_i}{x_1},\qquad i=2,\ldots,k.
\]
\begin{lemma}\label{lemma:ratios}
For a channel with $k>1$, $\xi=1$, and $\gamma=\infty$, two vectors $\mathbf{x}$ and $\mathbf{x}'$ are distinguishable if and only if their ratio vectors $\mathbf{u}$ and $\mathbf{u}'$ are distinct.
\end{lemma}
\begin{proof}
Suppose $\frac{x_i}{x_1}=\frac{x'_i}{x'_1}$ for $i=2,\ldots,k$. By letting $T=T_1x'_1$ and $T'=T_1x_1$, we see that \eqref{eq:def_indistinguishable} holds.
For the other direction, suppose there exist $T,T'$ s.t. $Tx_i=T'x'_i$ for $i=1,\ldots,k$. Dividing by $Tx_1$ or equivalently $T'x'_1$ implies the appropriate ratio vectors are equal.
\end{proof}
According to Lemma~\ref{lemma:ratios}, we can construct an optimal code by taking the maximal number of input vectors with distinct ratio vectors.
\begin{theorem}\label{thm:no_jitter_unbounded_drift}
The following code is an optimal zero-error code for a channel with $k>1$, $\xi=1$, and $\gamma=\infty$:
\begin{equation}
\mathcal{C}^*_{1,\infty}=\big\{\mathbf{x}\in\mathcal{X}:\ \gcd(\mathbf{x})=1\big\},
\label{eq:optimal_code_no_jitter_unbounded_drift}
\end{equation}
where $\gcd(\mathbf{x})=\gcd(x_1,\ldots,x_k)$ is the greatest common divisor of $(x_1,\ldots,x_k)$, i.e. it is the largest integer $d$ s.t. $d \mid x_i$ for all $i=1,\ldots,k$.
\end{theorem}
Any vector $\mathbf{x}$ can be divided by its gcd to obtain a vector with the same ratios vector $\mathbf{u}$ and gcd 1.
Therefore, this code is a maximal set of codewords with distinct ratio vectors.
The receiver can decode by computing the ratios of the received signal $\frac{Y_i}{Y_1}=\frac{TX_i}{TX_1}=U_i$.
We note that there are $\binom{M}{k}$ input vectors in $\mathcal{X}$, so $\mathcal{C}^*_{1,\infty}$ can be constructed in $O(M^k)$ time by an exhaustive search,
%\footnote{Typical numbers are, for example, $k=2$ and $M=64$, for which $|\mathcal{X}|=2016$.} 
which, to the best of our knowledge, is the best that can be done.

We can compare the optimal code here with the 64-PPM scheme in \cite{tabesh2015power}, by setting $k=2$ and $M=65$. While the scheme in \cite{tabesh2015power} consisted of 3 pulses, the first one is used only to mark the beginning of a frame\footnote{While in our model we assume the receiver knows the exact 
% time communication starts, 
{starting time of communication,}
in practice there may be an unknown timing offset (caused by e.g. unknown time-of-flight). This adverse effect is not studied in this work.} and hence will not be counted for the purpose of this comparison.
Therefore discarding the first pulse in \cite{tabesh2015power}, the codebook contains codewords for which $x_1=1$ and $1\leq x_2\leq M-1$, which has a rate of 6 bits per frame.
%and the receiver uses its observation $Y_1=T$ to measure the clock cycle and decode $x_2$. This gives a rate of 6 bits per frame.
Computing $\mathcal{C}^*_{1,\infty}$ from~\eqref{eq:optimal_code_no_jitter_unbounded_drift} yields $R=10.35$ bits per frame. In the next section we will see how this can be improved even further by taking into account the bounds on the clock $T$.

\begin{proof}[Proof of Theorem~\ref{thm:no_jitter_unbounded_drift}]
To show that this is a zero-error code, let $\mathbf{x},\mathbf{x}'\in\mathcal{C}^*_{1,\infty}$ be two distinct codewords, and suppose their ratio vectors are equal: 
\[\frac{x_i}{x_1}=\frac{x'_i}{x'_1},\qquad i=2,\ldots,k.\]
Then necessarily $x_1\neq x'_1$, otherwise the codewords are not distinct. Let $\frac{m}{n}$ be the reduced fraction of $\frac{x'_1}{x_1}$, i.e. $\frac{m}{n}=\frac{x'_1}{x_1}$ and $\gcd(m,n)=1$. Then for each $i=1,\ldots,k$:
\[
x'_i=\frac{x'_1}{x_1}x_i
=\frac{m}{n}x_i,
\]
and since $x'_i$ is an integer, $n$ must divide $x_i$ for all $i=1,\ldots,k$.
Since by assumption $\gcd(\mathbf{x})=1$, we must have $n=1$. This implies $x'_i=mx_i$ for all $i$, where $m$ is an integer greater than 1, which means $\gcd(\mathbf{x}')=m>1$. This is a contradiction since $\mathbf{x}'\in\mathcal{C}^*_{1,\infty}$. Hence $\frac{x_i}{x_1}\neq\frac{x'_i}{x'_1}$ for some $i$, and by Lemma~\ref{lemma:ratios} the code is zero-error.

Next, we claim that $\mathcal{C}^*_{1,\infty}$ is optimal by showing that any other zero-error code $\mathcal{C}$ must have at most as many codewords as $\mathcal{C}^*_{1,\infty}$. To this end, construct the code $\tilde{\mathcal{C}}$ from $\mathcal{C}$ by modifying each codeword as follows:
\[
\tilde{x}_i=\frac{x_i}{\gcd(\mathbf{x})},\qquad i=1,\ldots,k,
\]
or in short $\tilde{\mathbf{x}}=\frac{\mathbf{x}}{\gcd(\mathbf{x})}$.
The new codewords all have $\gcd(\tilde{\mathbf{x}})=1$, which implies $\tilde{\mathcal{C}}\subseteq\mathcal{C}^*_{1,\infty}$. By the previous arguments made for $\mathcal{C}^*_{1,\infty}$, the new code $\tilde{\mathcal{C}}$ is zero-error. Moreover, no two codewords in $\mathcal{C}$ map to the same codeword in $\tilde{\mathcal{C}}$; this follows from Lemma~\ref{lemma:ratios} and because $\mathcal{C}$ is zero-error. Therefore $|\tilde{\mathcal{C}}|=|\mathcal{C}|$, which implies $|{\mathcal{C}}|\leq|\mathcal{C}^*_{1,\infty}|$.
% To show that $\mathcal{C}^*_{1,\infty}$ is optimal, take any zero-error code $\mathcal{C}$ and divide each codeword $\mathbf{x}\in\mathcal{C}$ by $\gcd(\mathbf{x})$. This will give a subset of $\mathcal{C}^*_{1,\infty}$ without reducing the number of codewords (by Lemma~\ref{lemma:ratios}), which implies $|\mathcal{C}|\leq |\mathcal{C}^*_{1,\infty}|$.
\end{proof}

\subsection{No Jitter ($\xi=1$) and Bounded Clock Drift ($\gamma<\infty$)}
\label{subsec:no_jitter_bounded_drift}

%We construct an optimal code for the case of $\gamma<\infty$ by building on the results of the previous section and on ideas from \cite{yeung2009reliable}.
Note that the codewords in $\mathcal{C}^*_{1,\infty}$, defined in~\eqref{eq:optimal_code_no_jitter_unbounded_drift}, are distinguishable also when $\gamma<\infty$. Therefore, to construct an optimal code for the current channel, it is enough to add appropriate codewords to $\mathcal{C}^*_{1,\infty}$.
More specifically, for any $\mathbf{x}\in\mathcal{C}^*_{1,\infty}$, one can add its multiples $d\cdot\mathbf{x}=(dx_1,\ldots,dx_k)$, while carefully choosing $d$ so that the codewords are distinguishable.
For this purpose, take any vector $\mathbf{x}\in\mathcal{X}$ with $\gcd(\mathbf{x})=1$, and construct the set $\mathcal{L}_{\mathbf{x}}^\gamma$ as follows:
\begin{enumerate}
\item Start with $\mathcal{L}_{\mathbf{x}}^\gamma=\{\mathbf{x}\}$ and let $d_1=1$.
\item Given $d_{i-1}$, let $d_i$ be the smallest integer such that $d_i/d_{i-1}>\gamma$, i.e. $d_i=\lfloor \gamma d_{i-1}+1\rfloor$.
\item If $d_i\mathbf{x}\in\mathcal{X}$, add it to $\mathcal{L}^\gamma_{\mathbf{x}}$ and repeat step 2. Otherwise, stop the construction.
\end{enumerate}
Observe that all vectors in $\mathcal{L}^\gamma_{\mathbf{x}}$ are distinguishable. To see this, take two vectors, $d\mathbf{x}$ and $d'\mathbf{x}$, where $d,d'$ are two distinct integers. If they are indistinguishable, then there exist $T,T'\in[T_1,T_2]$ s.t. $Tdx_i=T'd'x_i$ for $i=1,\ldots,k$, where $T_2/T_1=\gamma$. This implies $d/d' = T'/T \in [\gamma^{-1},\gamma]$. However, by construction $d,d'$ must satisfy $d/d'>\gamma$ or $d'/d>\gamma$. This is a contradiction, hence they must be distinguishable.
% It can be seen that all vectors in $\mathcal{L}^\gamma_{\mathbf{x}}$ are distinguishable. 
Moreover, this set is ``maximal'' in the sense that it contains the maximal number of distinguishable vectors of the form $d\mathbf{x}$ for some integer $d\geq1$.

In the following theorem, we construct an optimal code by taking a union of all the sets $\mathcal{L}^\gamma_{\mathbf{x}}$ for all vectors $\mathbf{x}\in\mathcal{X}$ with $\gcd(\mathbf{x})=1$, which are exactly the codewords in $\mathcal{C}_{1,\infty}^*$ defined in Theorem~\ref{thm:no_jitter_unbounded_drift}.
\begin{theorem}\label{thm:no_jitter_bounded_drift}
An optimal zero-error code for a channel with $\xi=1$ and $\gamma<\infty$ is given by
\begin{equation}
\mathcal{C}_{1,\gamma}^*=\bigcup_{\mathbf{x}\in\mathcal{C}_{1,\infty}^*}\mathcal{L}^\gamma_{\mathbf{x}},
\label{eq:optimal_code_no_jitter_bounded_drift}
\end{equation}
where $\mathcal{C}_{1,\infty}^*$ is given by \eqref{eq:optimal_code_no_jitter_unbounded_drift}.
\end{theorem}
In order to decode, the receiver first computes the ratios vector of the output. This uniquely identifies a vector $\mathbf{x}$ with $\gcd(\mathbf{x})=1$, or equivalently a set $\mathcal{L}_{\mathbf{x}}^{\gamma}$. Then the correct codeword in $\mathcal{L}_{\mathbf{x}}^{\gamma}$ can be decoded from any single run $Y_i$.

%Note that the sets $\mathcal{L}^\gamma_{\mathbf{x}}$ are disjoint, and codewords from different sets are distinguishable at the receiver from arguments of the previous section, hence this is a zero-error code.
%Indeed, this is also an optimal code due to the optimality of the construction $\mathcal{L}^\gamma_{\mathbf{x}}$.

Equipped with this theorem, we compute the optimal code when the clock cycle is bounded between 4 and 7 ns, which are the actual system parameters in \cite{tabesh2015power}. The clock drift parameter is $\gamma=1.75$, which yields a rate of 10.76 bits per frame. 
It is interesting to note that, while this is an improvement over the code for $\gamma=\infty$, it is not particularly significant. Therefore, at least in this case, the fact that the clock drift is bounded does not provide a meaningful gain to capacity.
%This suggests that designing a code for $\gamma=\infty$ may be close enough to optimality in some cases, even if $\gamma<\infty$.
Finally, note that the best rate that can be achieved, even without clock drift, is 11.02 bits per frame. This is obtained by the optimal code with $\binom{M}{k}=\binom{65}{2}$ codewords.

\begin{proof}[Proof of Theorem~\ref{thm:no_jitter_bounded_drift}]
From arguments made in the previous section and by the construction of $\mathcal{L}_{\mathbf{x}}^{\gamma}$, it follows that $\mathcal{C}^*_{1,\gamma}$ is zero-error.
To show that it is an optimal code, we take an arbitrary zero-error code $\mathcal{C}$ and construct another code $\tilde{\mathcal{C}}$. Specifically, for each codeword $\mathbf{x}\in\mathcal{C}$, let $d=\gcd(\mathbf{x})$ and consider the vector 
\[\frac{\mathbf{x}}{d}=(\frac{x_1}{d},\ldots,\frac{x_k}{d})\in\mathcal{C}^*_{1,\infty}.\]
Let $\tilde{d}$ be the largest integer such that $\tilde{d}\leq d$ and
\[
\tilde{d}\frac{\mathbf{x}}{d}=
\left(\tilde{d}\frac{x_1}{d},\ldots,\tilde{d}\frac{x_k}{d}\right)\in\mathcal{L}^\gamma_{\mathbf{x}/d}.
\]
We map $\mathbf{x}$ to $\tilde{\mathbf{x}}=\tilde{d}\frac{\mathbf{x}}{d}$. The set of all vectors $\tilde{\mathbf{x}}$ constitutes the new code $\tilde{\mathcal{C}}$.

Clearly $\tilde{\mathcal{C}}\subseteq\mathcal{C}^*_{1,\gamma}$. It remains to show $|\tilde{\mathcal{C}}|=|\mathcal{C}|$, i.e. no two codewords in $\mathcal{C}$ map to the same codeword in $\tilde{\mathcal{C}}$. For this purpose, let $\mathbf{x}, \mathbf{x}'\in\mathcal{C}$ be two distinct codewords, and assume they map to the same codeword $\tilde{\mathbf{x}}\in\tilde{\mathcal{C}}$. Let $d=\gcd(\mathbf{x})$ and $d'=\gcd(\mathbf{x}')$.
First, notice that necessarily $\frac{\mathbf{x}}{d}=\frac{\mathbf{x}'}{d'}$, otherwise they cannot map to the same $\tilde{\mathbf{x}}$. Denote $\bar{\mathbf{x}}=\frac{\mathbf{x}}{d}=\frac{\mathbf{x}'}{d'}$.
Then, we have $\tilde{\mathbf{x}}=\tilde{d}\bar{\mathbf{x}}$, where $\tilde{d}$ is the largest integer $\tilde{d}\leq d$ and $\tilde{d}\leq d'$ s.t. $\tilde{d}\bar{\mathbf{x}}\in\mathcal{L}^\gamma_{\bar{\mathbf{x}}}$.
Since $\mathbf{x},\mathbf{x}'$ are distinct, we can assume without loss of generality $d<d'$.
By construction of $\mathcal{L}^\gamma_{\bar{\mathbf{x}}}$, we must have $\frac{d'}{\tilde{d}}\leq\gamma$, otherwise there must be another integer $q\leq d'$ s.t. $q\bar{\mathbf{x}}\in\mathcal{L}^\gamma_{\bar{\mathbf{x}}}$ and $\frac{q}{\tilde{d}}>\gamma$, in contradiction to the fact that $\tilde{d}$ is the largest such integer with $\tilde{d}\leq d'$. Along with the inequality $\tilde{d}\leq d$, it follows that $\frac{d'}{d}\leq\gamma$. This, in turn, implies that $\mathbf{x},\mathbf{x}'$ are indistinguishable, which contradicts the assumption that $\mathcal{C}$ is zero-error.
\end{proof}

\subsection{Jitter ($\xi>1$) and No Clock Drift ($\gamma=1$)}
\label{subsec:jitter_no_clock_drift}

When $\gamma=1$, the problem reduces to the one studied in~\cite{yeung2009reliable}.
% The following lemma follows directly from \cite{yeung2009reliable} or from arguments similar to the previous sections, and will be stated without proof.
Nevertheless, we provide here the code construction and proofs for completeness.

\begin{lemma}\label{lemma:distinguishable_jitter_no_drift}
For a channel with $\xi>1$ and $\gamma=1$, two input vectors $\mathbf{x},\mathbf{x}'$ are distinguishable if and only if there is an index $1\leq i\leq k$ such that $x_i/x'_i>\xi$ or $x'_i/x_i>\xi$.
\end{lemma}
\begin{proof}
We can assume without loss of generality that $T=1$.
Then, two input vectors $\mathbf{x}$, $\mathbf{x}'$ are indistinguishable if and only if there exist $Z_1,\ldots,Z_k,Z'_1,\ldots,Z'_k\in [a,b]$ such that
$Z_ix_i=Z'_ix'_i$, or equivalently $\frac{x_i}{x'_i}=\frac{Z'_i}{Z_i}$, for every $i=1,\ldots,k$.
This, in turn, holds if and only if $\xi^{-1}\leq\frac{x_i}{x'_i}\leq\xi$ for all $i$, completing the proof.
\end{proof}
Note that now, since jitter can alter each run independently, there needs to be at least one run that is distinct (up to ``stretching'' or ``squeezing'' by $\xi$) between two input vectors in order for them to be distinguishable.
%This is, in general, a stricter requirement than the previous sections, where two codewords are indistinguishable only if they differ by scaling of the entire codeword up to $\gamma$.
It is, in general, harder to distinguish between vectors corrupted by timing jitter as compared to clock drift.
For example, for $k=2$, $\gamma=1$, and $\xi=2$, the following vectors are all indistinguishable: $(1,1)$, $(1,2)$, $(2,1)$, $(2,2)$. On the other hand, if $\xi=1$ and $\gamma=2$, only $(1,1)$ and $(2,2)$ are indistinguishable, while $(1,1)$ and $(1,2)$ are distinguishable for example.

Similarly to the construction in the previous section, we construct the set $\mathcal{L}^\xi_{1}$:
\begin{enumerate}
\item Start with $\mathcal{L}^\xi_1=\{1\}$ and let $l_1=1$.
\item Given $l_{i-1}$, set $l_i=\lfloor\xi l_{i-1}+1\rfloor$, which is the smallest integer s.t. $l_i/l_{i-1}>\xi$.
\item If $l_i\leq M$, add it to $\mathcal{L}_1^\xi$ and repeat step 2. Otherwise, stop the construction.
\end{enumerate}

An optimal code can be constructed by allowing each codeword to contain runs only from $\mathcal{L}_1^\xi$.
The following theorem is similar to \cite[Theorem~1]{yeung2009reliable}.
\begin{theorem}\label{thm:jitter_no_clock_drift}
An optimal zero-error code for a channel with $\xi>1$ and $\gamma=1$ is given by
\begin{equation}
\mathcal{C}^*_{\xi,1}=\big\{\mathbf{x}\in\mathcal{X}:\ 
x_i\in\mathcal{L}_1^\xi,\ i=1,\ldots,k\big\}.
\label{eq:optimal_code_jitter_no_drift}
\end{equation}
\end{theorem}
% The proof is similar to the proof of Theorem~\ref{thm:no_jitter_bounded_drift} or \cite[Theorem~1]{yeung2009reliable}, and is omitted here.
\begin{proof}
Observe that this is a zero-error code, since for any two distinct codewords $\mathbf{x},\mathbf{x}'\in\mathcal{C}^*_{\xi,1}$ there is at least one $i\in\{1,\ldots,k\}$ for which $x_i\neq x'_i$. Since both $x_i$ and $x'_i$ are in $\mathcal{L}_1^{\xi}$, this implies that either $x_i/x'_i>\xi$ or $x'_i/x_i>\xi$, which, according to Lemma~\ref{lemma:distinguishable_jitter_no_drift}, means that $\mathbf{x}$ and $\mathbf{x}'$ are distinguishable.

Next, we show that this is an optimal code by modifying an arbitrary zero-error code $\mathcal{C}$ in a similar manner to the proof of Theorem~\ref{thm:no_jitter_bounded_drift}.
% Specifically, we construct the code $\tilde{\mathcal{C}}$ by applying the following transformation for each codeword $\mathbf{x}\in\mathcal{C}$:
% let $\tilde{\mathbf{x}}=(\tilde{x}_1,\ldots,\tilde{x}_k)$, where $\tilde{x}_i$ is the largest element in $\mathcal{L}_1^\xi$ such that $\tilde{x}_i\leq x_i$.
Specifically, we construct the code $\tilde{\mathcal{C}}$ by {mapping each codeword $\mathbf{x}\in\mathcal{C}$ to the codeword $\tilde{\mathbf{x}}=(\tilde{x}_1,\ldots,\tilde{x}_k)$, where $\tilde{x}_i$ is the largest element in $\mathcal{L}_1^\xi$ such that $\tilde{x}_i\leq x_i$.}

Clearly $\tilde{\mathcal{C}}\subset\mathcal{C}^*_{\xi,1}$, therefore $|\tilde{\mathcal{C}}|\leq|\mathcal{C}^*_{\xi,1}|$.
Then it remains to show that we do not lose anything by modifying $\mathcal{C}$ to $\tilde{\mathcal{C}}$, i.e. no two codewords in $\mathcal{C}$ map to the same codeword in $\tilde{\mathcal{C}}$.
This will imply $\mathcal{C}$ has at most as many codewords as $\mathcal{C}^*_{\xi,1}$, which will conclude the proof that $\mathcal{C}^*_{\xi,1}$ is an optimal code.

Let $\mathbf{x},\mathbf{x}'\in\mathcal{C}$ be two distinct codewords, and assume they are mapped to the same codeword $\tilde{\mathbf{x}}\in\tilde{\mathcal{C}}$.
Therefore for every $i\in\{1,\ldots,k\}$, the element $\tilde{x}_i\in\mathcal{L}_1^{\xi}$ is the maximal such that $\tilde{x}_i\leq x_i$ and $\tilde{x}_i\leq x'_i$.
By construction of $\mathcal{L}_1^{\xi}$, it follows that $x_i,x'_i<\lfloor \xi\tilde{x}_i+1\rfloor$, which implies $\tilde{x}_i\leq x_i,x'_i\leq \xi\tilde{x}_i$.
Hence, $\xi^{-1}\leq\frac{x_i}{x'_i}\leq \xi$ for every $1\leq i\leq k$, and it follows from Lemma~\ref{lemma:distinguishable_jitter_no_drift} that $\mathbf{x}$ and $\mathbf{x}'$ are indistinguishable, which is a contradiction to the assumption that $\mathcal{C}$ is zero-error.
\end{proof}

The nature of jitter requires different coding and decoding techniques as compared to clock drift. When only clock drift is present, the receiver needs to wait for the entire signal before it can decode (which is done by computing the ratios vector).
When there is jitter but no clock drift, i.e. the clock cycle is known exactly at the receiver, the receiver can decode each run independently, and does not have to wait for the entire output vector.
As will be seen in the following section, this poses an interesting challenge when both jitter and clock drift occur.

\subsection{Jitter ($\xi>1$) and Unbounded Clock Drift ($\gamma=\infty$)}
\label{subsec:jitter_unbounded_drift}

We solve this only for the case of $k=2$. In the following lemma, we state a necessary and sufficient condition for two input vectors to be distinguishable at the receiver.
\begin{lemma}\label{lemma:distinguishable_jitter_unbounded_drift}
A pair of input vectors $(x_1,x_2)$ and $(x'_1,x'_2)$ are distinguishable for a channel with $k=2$, $\xi>1$, and $\gamma=\infty$, if and only if
$\frac{x_2}{x_1} > \xi^2 \frac{x'_2}{x'_1}$
or
$\frac{x'_2}{x'_1} > \xi^2\frac{x_2}{x_1}$.
\end{lemma}

Intuitively, when the clock drift is unbounded, distinguishable vectors must have distinct ratios $\frac{x_2}{x_1}$, hence vectors can be equally represented by their appropriate ratio.
However, when jitter corrupts the signal, the numerator and the denominator can ``stretch'' or ``squeeze'' independently, by a factor $\xi$ each. Together, the ratio can change by up to a factor of $\xi^2$.

%\todo[inline]{Ayfer: I am not sure about the following paragraph, should we remove it?}
%This line of reasoning is difficult to extend to $k>2$. This is because the ratios vector computed at the receiver for a $k$-tuple $\mathbf{x}$ is given by
%$\left(\frac{y_2}{y_1},\ldots,\frac{y_k}{y_1}\right)=
%\left(\frac{Z_2x_2}{Z_1x_1},\ldots,\frac{Z_kx_k}{Z_1x_1}\right)$,
%and while the numerators can all change independently, the denominator is the same for all elements of the tuple.

\begin{proof}[Proof of Lemma~\ref{lemma:distinguishable_jitter_unbounded_drift}]
We will show that if two vectors satisfy $\frac{x_2}{x_1}\leq\xi^2\frac{x'_2}{x'_1}$ and $\frac{x'_2}{x'_1}\leq\xi^2\frac{x_2}{x_1}$ then they are indistinguishable.
For this purpose, we need to find $T,T'\in[T_1,\infty)$ and $Z_1,Z_2,Z'_1,Z'_2\in[a,b]$, where $a,b,T_1>0$ and $b/a=\xi$, such that \eqref{eq:def_indistinguishable} holds, i.e. $TZ_1x_1=T'Z'_1x'_1$ and $TZ_2x_2=T'Z'_2x'_2$.

Observe that $\xi^{-2}\leq\frac{x_2x'_1}{x_1x'_2}\leq\xi^2$.
Let $Z_1,Z_2,Z'_1,Z'_2$ be such that $\frac{Z_1Z'_2}{Z_2Z'_1}=\frac{x_2x'_1}{x_1x'_2}$; these exist since 
\[\xi^{-2}=\frac{a^2}{b^2}\leq\frac{Z_1Z'_2}{Z_2Z'_1}\leq\frac{b^2}{a^2}=\xi^2.\]
Having fixed $Z_1,Z_2,Z'_1,Z'_2$, find $T,T'$ such that $\frac{T'}{T}=\frac{Z_1x_1}{Z'_1x'_1}$.
This is possible since the ratio $\frac{T'}{T} $ can take any positive number. Now, we have
\[
\frac{T'Z'_2x'_2}{TZ_2x_2}
=\frac{Z_1x_1}{Z'_1x'_1}\frac{Z'_2x'_2}{Z_2x_2}
=1,
\]
implying that $(x_1,x_2)$ and $(x'_1,x'_2)$ are indistinguishable.

The other direction, namely that if $(x_1,x_2)$ and $(x'_1,x'_2)$ are indistinguishable then $\xi^{-2}\frac{x'_2}{x'_1}\leq\frac{x_2}{x_1}\leq\xi^2\frac{x'_2}{x'_1}$, follows by repeating the previous arguments in the reverse direction.
\end{proof}

Since distinguishable codewords must have distinct ratios (whether jitter is present or not), we can, without loss of generality, take only codewords for which $\gcd(x_1,x_2)=1$.
Hence we can construct an optimal code for this channel by taking a subset of the optimal code for the channel without jitter, that is $\mathcal{C}^*_{\xi,\infty}\subseteq\mathcal{C}^*_{1,\infty}$.
Since the ratios $\frac{x_2}{x_1}$ of all codewords in $\mathcal{C}^*_{1,\infty}$ are distinct, we define the following set of fractions:
\begin{align*}
\mathcal{U}
&=\big\{\tfrac{x_2}{x_1}:\ (x_1,x_2)\in\mathcal{C}^*_{1,\infty}\big\}.
%\\
%&=\left\{\frac{x_2}{x_1}:\ 
%(x_1,x_2)\in\mathcal{X},\ \gcd(x_1,x_2)=1
%\right\}.
\end{align*}
There is a one-to-one mapping between $\mathcal{U}$ and $\mathcal{C}_{1,\infty}^*$.
Using the set $\mathcal{U}$, we construct an optimal code $\mathcal{C}^*_{\xi,\infty}$ by means of the following algorithm:
\begin{enumerate}
\item Start with $\mathcal{C}_{\xi,\infty}^*=\{(M-1,1)\}$ and let $u_1=\frac{1}{M-1}$, which is the smallest element in $\mathcal{U}$.
\item Given $u_{i-1}$, {consider the set of all elements $u\in\mathcal{U}$ s.t. $u>\xi^2 u_{i-1}$, or in other words, the set $\mathcal{U}\cap(\xi^2 u_{i-1},\infty)$, where $(\xi^2 u_{i-1},\infty)$ denotes an open interval.}
% consider the set $\mathcal{U}\cap(\xi^2 u_{i-1},\infty)$. 
If the set is empty, stop the construction. Otherwise, let $(x_1,x_2)\in\mathcal{X}$ be the single vector s.t. $\gcd(x_1,x_2)=1$ and $\frac{x_2}{x_1}$ is the smallest element in $\mathcal{U}\cap(\xi^2 u_{i-1},\infty)$. Set $u_i=\frac{x_2}{x_1}$, add $(x_1,x_2)$ to $\mathcal{C}^*_{\xi,\infty}$, and repeat this step.
\end{enumerate}
This construction, while similar to the constructions $\mathcal{L}_{\mathbf{x}}^{\gamma}$ and $\mathcal{L}_1^\xi$ from the previous sections, operates on $\mathcal{U}$ which is a set of fractions, rather than on vectors or elements of vectors.
It is somewhat surprising that, given the different structure of $\mathcal{U}$ as compared to the set of vectors, this construction is indeed optimal, as stated formally in the following theorem.
\begin{theorem}
The code $\mathcal{C}_{\xi,\infty}^*$ obtained by the above construction is an optimal code for a channel with $k=2$, $\xi>1$, and $\gamma=\infty$.
\end{theorem}
\begin{proof}
By Lemma~\ref{lemma:distinguishable_jitter_unbounded_drift}, this code is zero-error.
To show that it is optimal, take any zero-error code $\mathcal{C}$ and construct the code $\tilde{\mathcal{C}}$ as follows:
for every codeword $(x_1,x_2)\in\mathcal{C}$, let $(\tilde{x}_1,\tilde{x}_2)$ be the codeword in $\mathcal{C}^*_{\xi,\infty}$ with the largest ratio $\frac{\tilde{x}_2}{\tilde{x}_1}$ such that $\frac{\tilde{x}_2}{\tilde{x}_1}\leq\frac{x_2}{x_1}$.
Clearly $\tilde{\mathcal{C}}\subseteq\mathcal{C}^*_{\xi,\infty}$, thus it remains to show that no two codewords in $\mathcal{C}$ map to the same codeword $(\tilde{x}_1,\tilde{x}_2)$, implying that $|\mathcal{C}|=|\tilde{\mathcal{C}}|$ and consequently $|\mathcal{C}|\leq|\mathcal{C}^*_{\xi,\infty}|$.

Let $(x_1,x_2)$ and $(x'_1,x'_2)$ be two distinct codewords in $\mathcal{C}$, and suppose they map to the same codeword $(\tilde{x}_1,\tilde{x}_2)\in\tilde{\mathcal{C}}$.
Since $\mathcal{C}$ is zero-error, $(x_1,x_2)$ and $(x'_1,x'_2)$ are distinguishable. Hence, by Lemma~\ref{lemma:distinguishable_jitter_unbounded_drift}, we can assume without loss of generality $\xi^2\frac{x_2}{x_1}<\frac{x'_2}{x'_1}$.
By definition of $\tilde{\mathcal{C}}$, we have $\frac{\tilde{x}_2}{\tilde{x}_1}\leq\frac{x_2}{x_1}$. It follows that 
$\xi^2\frac{\tilde{x}_2}{\tilde{x}_1}\leq\xi^2\frac{x_2}{x_1}<\frac{x'_2}{x'_1}$.
Then, since $\frac{x'_2}{x'_1}\in\mathcal{U}$, we have in particular $\frac{x'_2}{x'_1}\in\mathcal{U}\cap(\xi^2\frac{\tilde{x}_2}{\tilde{x}_1},\infty)$. Since $(\tilde{x}_1,\tilde{x}_2)\in\mathcal{C}^*_{\xi,\infty}$, there must be a codeword $(\hat{x}_1,\hat{x}_2)$ s.t. $\gcd(\hat{x}_1,\hat{x}_2)=1$ and $\xi^2\frac{\tilde{x}_2}{\tilde{x}_1}<\frac{\hat{x}_1}{\hat{x}_2}\leq\frac{x'_2}{x'_1}$.
This contradicts the assumption that $(\tilde{x}_1,\tilde{x}_2)$ is the codeword with the largest ratio $\frac{\tilde{x}_2}{\tilde{x}_1}$ s.t. $\frac{\tilde{x}_2}{\tilde{x}_1}\leq\frac{x'_2}{x'_1}$.
\end{proof}

\subsection{Jitter ($\xi>1$) and Bounded Clock Drift ($\gamma < \infty$)}
For this case, which is the most general, it is difficult to obtain exact characterization of the optimal code. Nevertheless, we provide an achievable zero-error code for the case of $k=2$.

Recall the construction of $\mathcal{L}^\gamma_{\mathbf{x}}$ from Section~\ref{subsec:no_jitter_bounded_drift}, defined for a tuple $\mathbf{x}$ with $\gcd(\mathbf{x})=1$.
Here, since each element $Y_i=TZ_ix_i$ can change by a factor of $\gamma\xi$, we take the same construction but with parameter $\gamma\xi$, namely $\mathcal{L}^{\gamma\xi}_{\mathbf{x}}$.
In the following theorem we build a (possibly suboptimal) zero-error code using these sets and the optimal code for $\gamma=\infty$ developed in the previous section.
\begin{theorem}\label{thm:lower_bound}
The following code is zero-error for a channel with $k=2$, $\xi>1$, and $\gamma<\infty$:
\begin{equation}
\mathcal{C}_{\xi,\gamma} = \bigcup_{\mathbf{x}\in\mathcal{C}^*_{\xi,\infty}} \mathcal{L}_{\mathbf{x}}^{\gamma\xi},
\end{equation}
where $\mathcal{L}_{\mathbf{x}}^{\gamma\xi}$ is defined in Section~\ref{subsec:no_jitter_bounded_drift} and $\mathcal{C}^*_{\xi,\infty}$ is defined in Section~\ref{subsec:jitter_unbounded_drift}.
\end{theorem}
\begin{proof}
By construction of the codebook $\mathcal{C}^*_{\xi,\infty}$ and
from the arguments of the previous section, it is clear that codewords from different $\mathcal{L}_{\mathbf{x}}^{\gamma\xi}$ are distinguishable (for any $\gamma$).
After decoding the ratio $\frac{x_2}{x_1}$ at the receiver, it then needs to decode one of the runs, say $x_1$. That is, we need to show that there are no $T,T'\in[T_1,T_2]$ and $Z_1,Z'_1\in[a,b]$ such that $TZ_1x_1=T'Z'_1x'_1$ for two distinct codewords, or equivalently $\frac{x_1}{x'_1}=\frac{T'Z'_1}{TZ_1}$.
By construction of $\mathcal{L}_{\mathbf{x}}^{\gamma\xi}$, we must have $\frac{x_1}{x'_1}>\gamma\xi$ or $\frac{x_1}{x'_1}<\gamma^{-1}\xi^{-1}$.
On the other hand, we have 
\[\gamma^{-1}\xi^{-1}=\frac{T_1a}{T_2b}\leq\frac{T'Z'_1}{TZ_1}\leq\frac{T_2b}{T_1a}=\gamma\xi.\]
Therefore all codewords in $\mathcal{C}_{\xi,\gamma}$ are distinguishable.
\end{proof}

\section{Numerical Results}
\label{sec:numerical_results}

Fig.~\ref{fig:clock_drift_capacity_plots} shows the zero-error capacity without jitter ($\xi=1$) as a function of the clock drift ratio $\gamma$.
It can be seen that the decrease in rate incurred by clock drift is rather small.

Fig.~\ref{fig:jitter_capacity_plots} shows the zero-error capacity and capacity lower bound as a function of jitter, for $\gamma=1$ (no clock drift), $\gamma=7/4$ (as in \cite{tabesh2015power}), and $\gamma=\infty$ (unbounded clock drift).
{Note that in general, a zero-error code designed for jitter $\xi$ will be zero-error for any $\xi'$ such that $\xi'<\xi$.
Hence the zero-error capacity should be a decreasing function of $\xi$.
This does not hold for the lower bound in Theorem~\ref{thm:lower_bound}, since the construction there is not necessarily optimal.
However, we can obtain a tighter lower bound for a given jitter parameter $\xi$ by using the largest codebook out of all the codebooks for $\xi'\geq\xi$: 
\[
R_{\xi,\gamma} = \max_{\xi'\geq\xi}\log|\mathcal{C}_{\xi',\gamma}|,
\]
which is depicted by the dashed curve.
}

Finally, Fig.~\ref{fig:capacity_frame_size_k2} and Fig.~\ref{fig:capacity_frame_size_k3} show the capacity as a function of the frame size $M$, without clock drift and with unbounded clock drift, for $k=2$ and $k=3$ pulses, respectively. This is compared to the naive scheme of \cite{tabesh2015power}, where the first pulse is used to learn the clock cycle duration $T$, and the remaining $k-1$ pulses can be allocated freely in the remaining $M-1$ bins, yielding a rate of $\log\binom{M-1}{k-1}$.

\begin{figure}
\centering
% \input{capacity_clock_drift.tex}
% This file was created by matplotlib2tikz v0.6.14.
\begin{tikzpicture}

% \definecolor{color0}{rgb}{0.12156862745098,0.466666666666667,0.705882352941177}

\begin{axis}[
xlabel={\small Clock drift ratio $\gamma$},
ylabel={\small $R$ [bits / frame]},
xmin=1, xmax=64,
xtick={1,2,4,8,16,32,64},
xticklabels={1,2,4,8,16,32,64},
ticklabel style = {font=\small},
ymin=10.3185272064338, ymax=11.0558840323901,
xmode=log,
tick align=outside,
tick pos=left,
x grid style={lightgray!92.02614379084967!black},
y grid style={lightgray!92.02614379084967!black}
]
\addplot [thick, black, forget plot]
table {%
1 11.0223678130285
1.04246576084112 11.0195907283579
1.08673486252606 11.0014081943928
1.1328838852958 10.9801395776392
1.18099266142953 10.9571020415623
1.23114441334492 10.9425145053392
1.2834258975629 10.9053870050181
1.33792755478611 10.871135184243
1.39474366635041 10.871135184243
1.45397251732031 10.8610869059954
1.5157165665104 10.7723145739217
1.58008262372675 10.7723145739217
1.64718203453515 10.7714894695006
1.71713087287551 10.7698378436294
1.79005014185594 10.7615512324445
1.86606598307361 10.7615512324445
1.94530989482457 10.7615512324445
2.02791895958006 10.5304063370991
2.11403608112276 10.5304063370991
2.20381023175322 10.5294305541462
2.29739670999407 10.5274770060604
2.39495740923786 10.5216004397237
2.49666109780322 10.5216004397237
2.60268371088387 10.5216004397237
2.71320865489534 10.515699838284
2.82842712474619 10.515699838284
2.9485384345822 10.515699838284
3.07375036257602 10.4460494067166
3.20427951035849 10.4460494067166
3.34035167771348 10.4419067045422
3.4822022531845 10.4419067045422
3.63007662126864 10.4419067045422
3.78423058690238 10.4419067045422
3.94493081797344 10.4419067045422
4.11245530662427 10.4167975276061
4.28709385014517 10.4146852358072
4.46914855228888 10.4146852358072
4.65893434587382 10.4146852358072
4.85677953758019 10.4146852358072
5.06302637588112 10.3869402453243
5.27803164309158 10.3869402453243
5.50216727255897 10.3858624006415
5.73582099206331 10.3858624006415
5.97939699453975 10.3858624006415
6.233316637284 10.3815429511846
6.49801917084988 10.3815429511846
6.77396249890022 10.3815429511846
7.06162397032524 10.3750394313469
7.361501204999 10.3750394313469
7.67411295460211 10.3750394313469
8 10.3706874068072
8.33972608672897 10.3706874068072
8.69387890020846 10.3706874068072
9.06307108236639 10.3641346550081
9.44794129143624 10.3641346550081
9.84915530675933 10.3641346550081
10.2674071805032 10.3619437737352
10.7034204382889 10.3619437737352
11.1579493308032 10.3619437737352
11.6317801385625 10.3619437737352
12.1257325320832 10.3619437737352
12.640660989814 10.3619437737352
13.1774562762812 10.3575520046181
13.7370469830041 10.3575520046181
14.3204011348476 10.3575520046181
14.9285278645889 10.3575520046181
15.5624791585966 10.3575520046181
16.2233516766405 10.3553510964248
16.9122886489821 10.3553510964248
17.6304818540258 10.3553510964248
18.3791736799526 10.3553510964248
19.1596592739029 10.3553510964248
19.9732887824258 10.3553510964248
20.8214696870709 10.3553510964248
21.7056692391628 10.3531468254981
22.6274169979695 10.3531468254981
23.5883074766576 10.3531468254981
24.5900029006082 10.3531468254981
25.6342360828679 10.3531468254981
26.7228134217078 10.3531468254981
27.857618025476 10.3531468254981
29.0406129701491 10.3531468254981
30.2738446952191 10.3531468254981
31.5594465437875 10.3531468254981
32.8996424529941 10.3520434257954
34.2967508011614 10.3520434257954
35.753188418311 10.3520434257954
37.2714747669906 10.3520434257954
38.8542363006415 10.3520434257954
40.5042110070489 10.3520434257954
42.2242531447326 10.3520434257954
44.0173381804718 10.3520434257954
45.8865679365065 10.3520434257954
47.835175956318 10.3520434257954
49.866533098272 10.3520434257954
51.9841533667991 10.3520434257954
54.1916999912017 10.3520434257954
56.4929917626019 10.3520434257954
58.892009639992 10.3520434257954
61.3929036368169 10.3520434257954
64 10.3520434257954
};
\end{axis}

\end{tikzpicture}
\caption{Zero-error rates for $k=2$ and $M=65$ without jitter ($\xi = 1$).}
\label{fig:clock_drift_capacity_plots}
\end{figure}
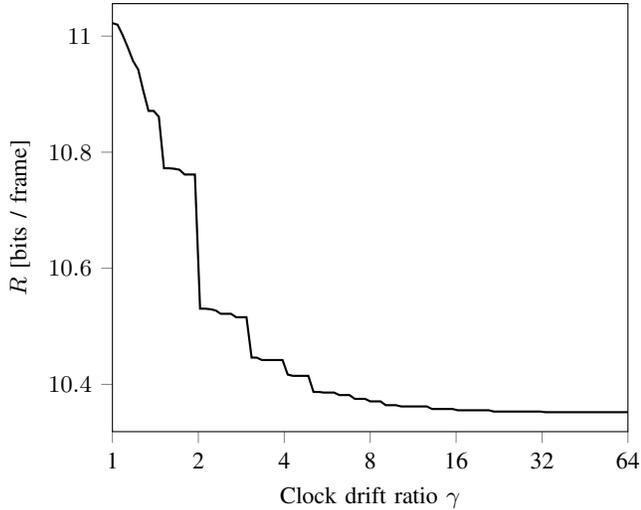

{
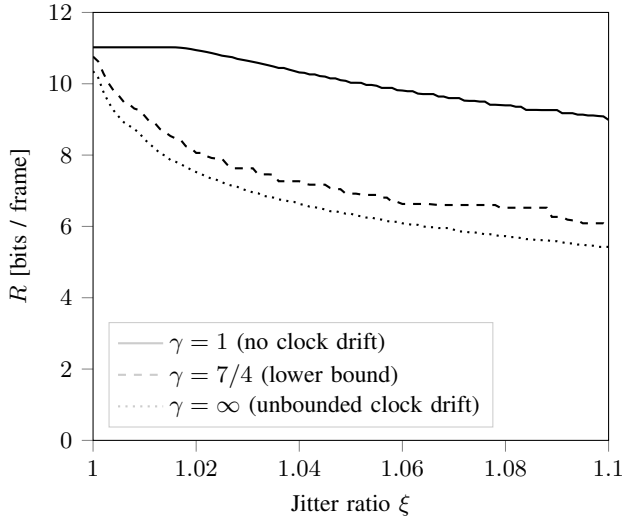
\begin{figure}
\centering
% \input{capacity_jitter.tex}
% This file was created by matplotlib2tikz v0.6.14.
\begin{tikzpicture}

% \definecolor{color1}{rgb}{1,0.498039215686275,0.0549019607843137}
% \definecolor{color0}{rgb}{0.12156862745098,0.466666666666667,0.705882352941177}
% \definecolor{color2}{rgb}{0.172549019607843,0.627450980392157,0.172549019607843}

\begin{axis}[
xlabel={\small Jitter ratio $\xi$},
ylabel={\small $R$ [bits / frame]},
xmin=1, xmax=1.1,
ymin=0, ymax=12,
tick align=outside,
tick pos=left,
ticklabel style = {font=\small},
x grid style={lightgray!92.02614379084967!black},
y grid style={lightgray!92.02614379084967!black},
legend style={font=\small, at={(0.03,0.03)}, anchor=south west, draw=white!80.0!black},
legend entries={{$\gamma = 1$ (no clock drift)},{$\gamma = 7 / 4$ (lower bound)},{$\gamma = \infty$ (unbounded clock drift)}},
legend cell align={left}
]
% \addlegendimage{no markers, thick}
% \addlegendimage{no markers, thick, dashed}
% \addlegendimage{no markers, thick, dotted}
\addplot [thick]
table {%
1 11.0223678130285
1.001 11.0223678130285
1.002 11.0223678130285
1.003 11.0223678130285
1.004 11.0223678130285
1.005 11.0223678130285
1.006 11.0223678130285
1.007 11.0223678130285
1.008 11.0223678130285
1.009 11.0223678130285
1.01 11.0223678130285
1.011 11.0223678130285
1.012 11.0223678130285
1.013 11.0223678130285
1.014 11.0223678130285
1.015 11.0223678130285
1.016 11.0209799389042
1.017 11.0098286173681
1.018 10.9943534368589
1.019 10.9715435539508
1.02 10.9425145053392
1.021 10.9188632372746
1.022 10.8917837032183
1.023 10.8610869059954
1.024 10.8265484872909
1.025 10.7879025593914
1.026 10.7665289085989
1.027 10.7448338374995
1.028 10.6969675262343
1.029 10.6706562491184
1.03 10.6438561897747
1.031 10.6147098441152
1.032 10.5849625007212
1.033 10.5517082616207
1.034 10.5186531556734
1.035 10.4838157772643
1.036 10.4429434958487
1.037 10.4429434958487
1.038 10.3976746329483
1.039 10.3586512008929
1.04 10.3174126137649
1.041 10.3060616894283
1.042 10.2620948453702
1.043 10.2620948453702
1.044 10.2155329997457
1.045 10.2008986050384
1.046 10.1510165388922
1.047 10.1510165388922
1.048 10.1006623390052
1.049 10.0821490413539
1.05 10.0279059965699
1.051 10.0279059965699
1.052 10.0251395622785
1.053 9.96722625883599
1.054 9.96722625883599
1.055 9.94397991434374
1.056 9.88417051910844
1.057 9.88417051910844
1.058 9.88417051910844
1.059 9.82017896241519
1.06 9.81057163474115
1.061 9.79116288855502
1.062 9.79116288855502
1.063 9.72451385311995
1.064 9.71080643369935
1.065 9.71080643369935
1.066 9.71080643369935
1.067 9.64205169292798
1.068 9.64205169292798
1.069 9.5980525001616
1.07 9.5980525001616
1.071 9.5980525001616
1.072 9.5274770060604
1.073 9.51963625284321
1.074 9.51963625284321
1.075 9.49984588708321
1.076 9.49984588708321
1.077 9.41151098801207
1.078 9.41151098801207
1.079 9.403012023575
1.08 9.39446269461032
1.081 9.39446269461032
1.082 9.35535109642481
1.083 9.35535109642481
1.084 9.2667865406949
1.085 9.2667865406949
1.086 9.2667865406949
1.087 9.26209484537018
1.088 9.26209484537018
1.089 9.26209484537018
1.09 9.26209484537018
1.091 9.17492568250068
1.092 9.17492568250068
1.093 9.17492568250068
1.094 9.13442632022093
1.095 9.13442632022093
1.096 9.11113567023471
1.097 9.11113567023471
1.098 9.09539702279256
1.099 9.08480838780436
1.1 8.98584193700334
};
\addplot [thick, dashed]
table {%
1 10.7615512324445
1.001 10.6257088430645
1.002 10.2957689344205
1.003 10.0168082876866
1.004 9.80413102118332
1.005 9.66355810421727
1.006 9.48381577726426
1.007 9.4262647547021
1.008 9.3037807481771
1.009 9.25974326369078
1.01 9.1006623390052
1.011 8.93073733756289
1.012 8.89481776330794
1.013 8.7279204545632
1.014 8.64024493622235
1.015 8.5352753766208
1.016 8.467605550083
1.017 8.4262647547021
1.018 8.24792751344359
1.019 8.19475685442225
1.02 8.06069593168755
1.021 8.06069593168755
1.022 7.97154355395077
1.023 7.91886323727459
1.024 7.90086680798075
1.025 7.90086680798075
1.026 7.73470962022584
1.027 7.62935662007961
1.028 7.62935662007961
1.029 7.62935662007961
1.03 7.62935662007961
1.031 7.62935662007961
1.032 7.4594316186373
1.033 7.4594316186373
1.034 7.4594316186373
1.035 7.45121111183233
1.036 7.2667865406949
1.037 7.2667865406949
1.038 7.2667865406949
1.039 7.2667865406949
1.04 7.2667865406949
1.041 7.19967234483636
1.042 7.16992500144231
1.043 7.16992500144231
1.044 7.16992500144231
1.045 7.16992500144231
1.046 7.07681559705083
1.047 7.04439411935845
1.048 7.04439411935845
1.049 6.91886323727459
1.05 6.91886323727459
1.051 6.91886323727459
1.052 6.89481776330794
1.053 6.88264304936184
1.054 6.88264304936184
1.055 6.88264304936184
1.056 6.8073549220576
1.057 6.8073549220576
1.058 6.68650052718322
1.059 6.68650052718322
1.06 6.62935662007961
1.061 6.62935662007961
1.062 6.62935662007961
1.063 6.62935662007961
1.064 6.62935662007961
1.065 6.62935662007961
1.066 6.61470984411521
1.067 6.59991284218713
1.068 6.59991284218713
1.069 6.59991284218713
1.07 6.59991284218713
1.071 6.59991284218713
1.072 6.59991284218713
1.073 6.59991284218713
1.074 6.59991284218713
1.075 6.59991284218713
1.076 6.59991284218713
1.077 6.59991284218713
1.078 6.59991284218713
1.079 6.52356195605701
1.08 6.52356195605701
1.081 6.52356195605701
1.082 6.52356195605701
1.083 6.52356195605701
1.084 6.52356195605701
1.085 6.52356195605701
1.086 6.52356195605701
1.087 6.52356195605701
1.088 6.52356195605701
1.089 6.2667865406949
1.09 6.2667865406949
1.091 6.2667865406949
1.092 6.18982455888002
1.093 6.16992500144231
1.094 6.16992500144231
1.095 6.08746284125034
1.096 6.08746284125034
1.097 6.08746284125034
1.098 6.08746284125034
1.099 6.08746284125034
1.1 5.97727992349992
};
\addplot [thick, dotted]
table {%
1 10.3520434257954
1.001 10.1686721181322
1.002 9.77643303244473
1.003 9.48179943166575
1.004 9.25029841790633
1.005 9.07414146275251
1.006 8.91587937883577
1.007 8.82336724004623
1.008 8.73131903102506
1.009 8.61102479730735
1.01 8.43879185257826
1.011 8.30833903013941
1.012 8.18487534290828
1.013 8.04984854945056
1.014 7.98299357469431
1.015 7.86418614465428
1.016 7.81378119121704
1.017 7.74819284958946
1.018 7.66533591718518
1.019 7.59245703726808
1.02 7.52356195605701
1.021 7.467605550083
1.022 7.39231742277876
1.023 7.33985000288462
1.024 7.3037807481771
1.025 7.23840473932508
1.026 7.17990909001493
1.027 7.13955135239879
1.028 7.10852445677817
1.029 7.04439411935845
1.03 7
1.031 6.95419631038687
1.032 6.93073733756289
1.033 6.85798099512757
1.034 6.83289001416474
1.035 6.8073549220576
1.036 6.75488750216347
1.037 6.74146698640115
1.038 6.70043971814109
1.039 6.68650052718322
1.04 6.62935662007961
1.041 6.59991284218713
1.042 6.55458885167764
1.043 6.53915881110803
1.044 6.5077946401987
1.045 6.4757334309664
1.046 6.4262647547021
1.047 6.4262647547021
1.048 6.39231742277876
1.049 6.35755200461808
1.05 6.35755200461808
1.051 6.3037807481771
1.052 6.28540221886225
1.053 6.24792751344359
1.054 6.24792751344359
1.055 6.20945336562895
1.056 6.18982455888002
1.057 6.14974711950468
1.058 6.12928301694497
1.059 6.12928301694497
1.06 6.08746284125034
1.061 6.06608919045777
1.062 6.04439411935845
1.063 6.04439411935845
1.064 6
1.065 5.97727992349992
1.066 5.97727992349992
1.067 5.95419631038687
1.068 5.95419631038687
1.069 5.95419631038687
1.07 5.90689059560852
1.071 5.85798099512757
1.072 5.85798099512757
1.073 5.83289001416474
1.074 5.83289001416474
1.075 5.8073549220576
1.076 5.78135971352466
1.077 5.78135971352466
1.078 5.75488750216347
1.079 5.7279204545632
1.08 5.7279204545632
1.081 5.70043971814109
1.082 5.70043971814109
1.083 5.6724253419715
1.084 5.64385618977472
1.085 5.64385618977472
1.086 5.61470984411521
1.087 5.61470984411521
1.088 5.61470984411521
1.089 5.58496250072116
1.09 5.58496250072116
1.091 5.55458885167764
1.092 5.52356195605701
1.093 5.52356195605701
1.094 5.49185309632967
1.095 5.49185309632967
1.096 5.4594316186373
1.097 5.4594316186373
1.098 5.4262647547021
1.099 5.4262647547021
1.1 5.4262647547021
};
\end{axis}

\end{tikzpicture}
\caption{Zero-error rates for $k=2$ and $M=65$ with clock drift and jitter.}
\label{fig:jitter_capacity_plots}
\end{figure}
}

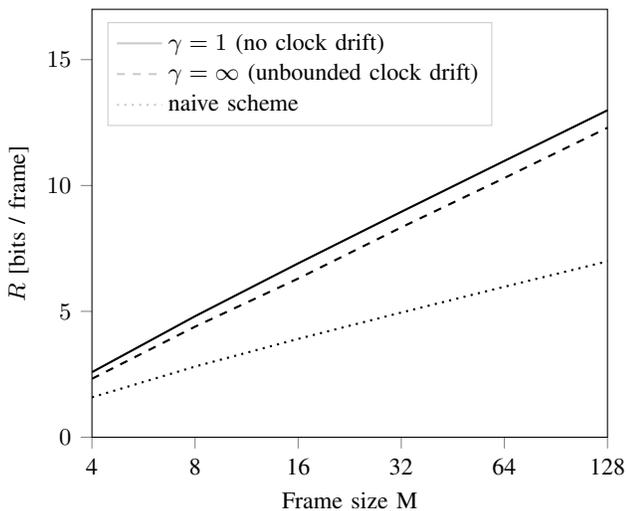
\begin{figure}
\centering
% This file was created by matplotlib2tikz v0.6.14.
\begin{tikzpicture}

% \definecolor{color1}{rgb}{1,0.498039215686275,0.0549019607843137}
% \definecolor{color0}{rgb}{0.12156862745098,0.466666666666667,0.705882352941177}
% \definecolor{color2}{rgb}{0.172549019607843,0.627450980392157,0.172549019607843}

\begin{axis}[
xlabel={\small Frame size M},
ylabel={\small $R$ [bits / frame]},
xmin=4, xmax=128,
ymin=0, ymax=17,
xmode=log,
tick align=outside,
tick pos=left,
xtick={4,8,16,32,64,128},
xticklabels={4,8,16,32,64,128},
ticklabel style = {font=\small},
x grid style={lightgray!92.02614379084967!black},
y grid style={lightgray!92.02614379084967!black},
legend entries={{$\gamma = 1$ (no clock drift)},{$\gamma = \infty$ (unbounded clock drift)},{naive scheme}},
legend cell align={left},
legend style={font=\small, at={(0.03,0.97)}, anchor=north west, draw=white!80.0!black}
]
% \addlegendimage{no markers, color0}
% \addlegendimage{no markers, color1}
% \addlegendimage{no markers, color2}
\addplot [thick]
table {%
4 2.58496250072116
8 4.8073549220576
16 6.90689059560852
32 8.95419631038687
64 10.9772799234999
128 12.9886846867722
};
\addplot [thick, dashed]
table {%
4 2.32192809488736
8 4.39231742277876
16 6.3037807481771
32 8.33539035469392
64 10.298062567719
128 12.2937590096679
};
\addplot [thick, dotted]
table {%
4 1.58496250072116
8 2.8073549220576
16 3.90689059560852
32 4.95419631038687
64 5.97727992349992
128 6.98868468677217
};
\end{axis}

\end{tikzpicture}
\caption{Zero-error rates without jitter ($\xi=1$) for $k=2$ as a function of the frame size $M$, with and without clock drift, and compared to the naive scheme of \cite{tabesh2015power}.}
\label{fig:capacity_frame_size_k2}
\end{figure}

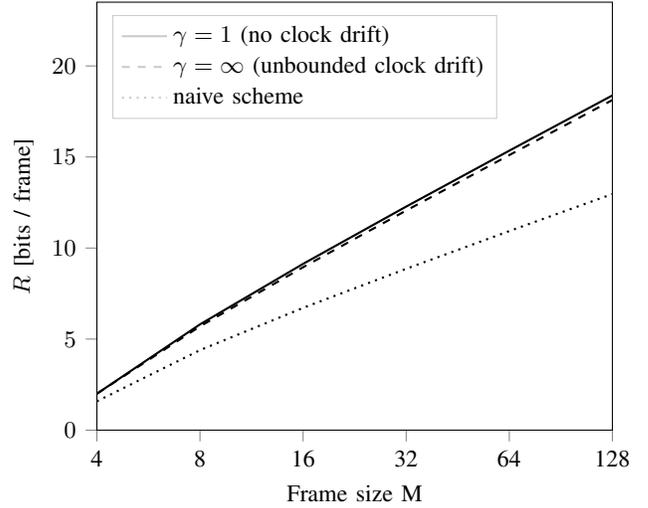
\begin{figure}
\centering
% This file was created by matplotlib2tikz v0.6.14.
\begin{tikzpicture}

% \definecolor{color1}{rgb}{1,0.498039215686275,0.0549019607843137}
% \definecolor{color0}{rgb}{0.12156862745098,0.466666666666667,0.705882352941177}
% \definecolor{color2}{rgb}{0.172549019607843,0.627450980392157,0.172549019607843}

\begin{axis}[
xlabel={\small Frame size M},
ylabel={\small $R$ [bits / frame]},
xmin=4, xmax=128,
ymin=0, ymax=23.5,
xmode=log,
tick align=outside,
tick pos=left,
xtick={4,8,16,32,64,128},
xticklabels={4,8,16,32,64,128},
ticklabel style = {font=\small},
x grid style={lightgray!92.02614379084967!black},
y grid style={lightgray!92.02614379084967!black},
legend entries={{$\gamma = 1$ (no clock drift)},{$\gamma = \infty$ (unbounded clock drift)},{naive scheme}},
legend cell align={left},
legend style={font=\small, at={(0.03,0.97)}, anchor=north west, draw=white!80.0!black}
]
% \addlegendimage{no markers, color0}
% \addlegendimage{no markers, color1}
% \addlegendimage{no markers, color2}
\addplot [thick]
table {%
4 2
8 5.8073549220576
16 9.12928301694497
32 12.2761244052742
64 15.3465137331656
128 18.3810021095509
};
\addplot [thick, dashed]
table {%
4 2
8 5.70043971814109
16 8.94544383637791
32 12.0590063952011
64 15.1034107601364
128 18.1282561805259
};
\addplot [thick, dotted]
table {%
4 1.58496250072116
8 4.39231742277876
16 6.71424551766612
32 8.86108690599539
64 10.9314762338868
128 12.9659646102721
};
\end{axis}

\end{tikzpicture}
\caption{Zero-error rates without jitter ($\xi=1$) for $k=3$ as a function of the frame size $M$, with and without clock drift, and compared to the naive scheme of \cite{tabesh2015power}.}
\label{fig:capacity_frame_size_k3}
\end{figure}

\section{Conclusion}
We introduced a model for communication with crystal-free radios, which includes two components of clock uncertainty: jitter and clock drift. The effects of slow clock drift suggest a new approach to designing codes for this type of radios. In particular, we show that estimating the clock cycle at the receiver may be suboptimal, and characterize the optimal code by considering ratios of runs, which are unaffected by clock drift.

When both jitter and clock drift are present, we find the capacity or achievable rate for a number of special cases. Characterizing the capacity and optimal zero-error codes for the case of general $(k, M, \xi, \gamma)$ is the subject of ongoing research.

\bibliographystyle{IEEEtran}
\bibliography{clock_drift_jitter}

\begin{IEEEbiography}[{\includegraphics[width=1in,height=1.25in,clip,
keepaspectratio]{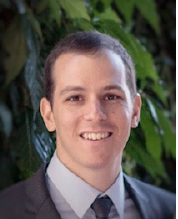}}]{Dor Shaviv}
(S'13) received the B.Sc. degrees (summa cum laude) in electrical engineering and physics and the M.Sc. degree in electrical engineering from the Technion---Israel Institute of Technology, Haifa, Israel, in 2007 and 2012, respectively.
He is currently a Ph.D. candidate in the Electrical Engineering Department at Stanford University.
From 2007 to 2013, he was a research and development engineer with the Israel Defense Forces.
He is a recipient of a Robert Bosch Stanford Graduate Fellowship.
\end{IEEEbiography}

\begin{IEEEbiography}[{\includegraphics[width=1in,height=1.25in,clip,
keepaspectratio]{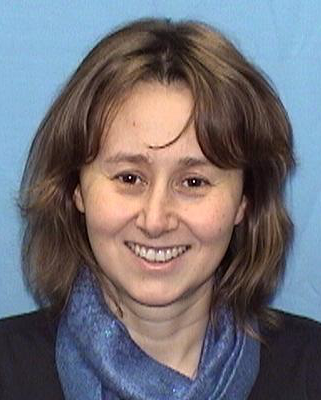}}]{Ayfer \"Ozg\"ur}
(M'06) received her B.Sc. degrees in electrical engineering and physics from Middle East Technical University, Turkey, in 2001 and the M.Sc. degree in communications from the same university in 2004. From 2001 to 2004, she worked as hardware engineer for the Defense Industries Development Institute in Turkey. She received her Ph.D. degree in 2009 from the Information Processing Group at EPFL, Switzerland. In 2010 and 2011, she was a post-doctoral scholar with the Algorithmic Research in Network Information Group at EPFL. She is currently an Assistant Professor in the Electrical Engineering Department at Stanford University. Her research interests include network communications, wireless systems, and information and coding theory. Dr.~\"Ozg\"ur received the EPFL Best Ph.D. Thesis Award in 2010 and a NSF CAREER award in 2013.
\end{IEEEbiography}

\begin{IEEEbiography}[{\includegraphics[width=1in,height=1.25in,clip,
keepaspectratio]{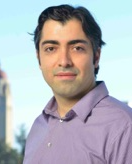}}]{Amin Arbabian}
Amin Arbabian (S'06–M'12–SM'17) received the Ph.D. degree in electrical engineering and computer science from the University of California at Berkeley, Berkeley, CA, in 2011. From 2007 and to 2008, he was part of the Initial Engineering Team, Tagarray, Inc., Palo Alto, CA, USA. He was with the Qualcomm’s Corporate Research and Development Division, San Diego, CA, in 2010. In 2012, he joined Stanford University, Stanford, CA, as an Assistant Professor of electrical engineering. His current research interests include mm-wave and high-frequency circuits and systems, imaging technologies, Internet-of-Everything devices including wireless power delivery techniques, and medical implants.

Dr. Arbabian was the recipient or co-recipient of the 2016 Stanford University Tau Beta Pi Award for Excellence in Undergraduate Teaching, the 2015 NSF CAREER Award, the 2014 DARPA Young Faculty Award including the Director’s Fellowship in 2016, the 2013 Hellman Faculty Scholarship, and best paper awards at the 2017 IEEE Biomedical Circuits and Systems Conference, the 2016 IEEE Conference on Biomedical Wireless Technologies, Networks, and Sensing Systems, the 2014 IEEE VLSI Circuits symposium, the 2013 IEEE International Conference on Ultra-Wideband, the 2010 IEEE Jack Kilby Award for Outstanding Student Paper at the International Solid-State Circuits Conference, and two time second place best student paper awards at 2008 and 2011 RFIC symposiums. He currently serves on the steering committee of RFIC Symposium, the technical program committees of RFIC symposium, ESSCIRC, and VLSI Circuits Symposium, and as an Associate Editor of the IEEE Solid-State Circuits Letters and the IEEE Journal of Electromagnetics, RF and Microwaves in Medicine and Biology.
\end{IEEEbiography}

\end{document}